\documentclass[aps,prl,twocolumn,amsfonts,amssymb,a4paper]{revtex4-1}
\usepackage{amsmath,amssymb,amsfonts,mathrsfs,eucal,amsthm}
\usepackage{dsfont}
\usepackage{graphicx,hyperref}
\usepackage[stable]{footmisc}
\usepackage{centernot}
\usepackage{harpoon}

\def\Tr{\mathrm{Tr}} \def\>{\rangle} \def\<{\langle}

 \def\test#1{\mathfrak{#1}}
\def\outx{\mathcal{X}} \def\outy{\mathcal{Y}}
\def\outt{\mathcal{T}} \def\outs{\mathcal{S}}

\def\mE{\mathcal{E}}
 \def\sH{\mathcal{H}} 
\def\sS{\mathsf{S}}

\def\openone{\mathds{1}}

\def\test#1{\boldsymbol{\mathscr{#1}}}

\def\linear{\mathsf{L}}

\def\prob{\mathscr{P}}

\def\mF{\mathcal{F}}
\def\game{\mathsf{G}_{\operatorname{nl}}}
\def\extgame{\mathsf{G}_{\operatorname{sq}}}

\def\sups{\supseteq_{\mathrm{ext}}}

\def\sups{\succcurlyeq_{\operatorname{sq}}}
\def\supnl{\succcurlyeq_{\operatorname{nl}}}

\def\succlosr{\dashrightarrow}

\def\conv{\mathsf{Co}}

\renewcommand{\qedsymbol}{\nobreak \ifvmode \relax \else
  \ifdim \lastskip<1.5em \hskip-\lastskip \hskip1.5em plus0em
  minus0.5em \fi \nobreak \vrule height0.75em width0.5em
  depth0.25em\fi}

\renewcommand{\ge}{\geqslant}
\renewcommand{\le}{\leqslant}
\renewcommand{\rho}{\varrho}

\newtheorem{corollary}{Corollary}

\newtheorem{proposition}{Proposition}

\theoremstyle{definition}
\newtheorem{definition}{Definition}

\begin{document}

\title{All Entangled Quantum States Are Nonlocal}

\author{Francesco Buscemi} \email{buscemi@iar.nagoya-u.ac.jp}
\affiliation{Institute for Advanced Research, Nagoya University,
  Chikusa-ku, Nagoya 464-8601, Japan}

\date{16 April 2012}

\begin{abstract}

  Departing from the usual paradigm of local operations and classical
  communication adopted in entanglement theory, here we study the
  interconversion of quantum states by means of local operations and
  shared randomness. A set of necessary and sufficient conditions for
  the existence of such a transformation between two given quantum
  states is given in terms of the payoff they yield in a suitable
  class of nonlocal games. It is shown that, as a consequence of our
  result, such a class of nonlocal games is able to witness quantum
  entanglement, however weak, and reveal nonlocality in any entangled
  quantum state. An example illustrating this fact is provided.

\end{abstract}

\maketitle


It is a fact that the outcomes of measurements performed on spatially
separated (i.e. non-communicating) quantum systems sometimes exhibit
correlations, which cannot be explained classically, in terms of
information shared beforehand. Such correlations, called
\emph{nonlocal}, are revealed by the violation of a suitable Bell
inequality~\cite{epr, bell}. Another peculiarly nonclassical feature
of quantum theory is the existence of \emph{quantum entanglement},
i.e. the property possessed by composite quantum systems whose joint
state cannot be written in product form (or, more generally, as a
mixture of states in product form). Even if nonlocality and
entanglement are indeed intimately related, it is nowadays widely
accepted that they are in fact two well distinct concepts: first of
all, because there exist entangled quantum states which behave
``locally'' in many aspects~\cite{werner:PRA, barret}; second, because
quantum states that appear to be ``maximally nonlocal'' are generally
not the ``maximally entangled'' ones~\cite{scarani:QIC}. Such a
quantitative distinction is made clear by looking upon nonlocality and
entanglement as two \emph{inequivalent resources}.

In the resource theory of quantum entanglement, the operational
paradigm is commonly known as \emph{local operations and classical
  communication} (LOCC)~\cite{bennett:PRL}: separated parties are only
allowed to exchange classical messages, while quantum operations
(i.e. preparation of quantum states, quantum measurements, etc.) can
only happen locally. In particular, quantum states cannot be directly
sent across separated locations. The LOCC paradigm, originally
formulated in order to describe the ``distant laboratories model'', is
nowadays generally accepted as the natural operational paradigm for
studying quantum entanglement as a resource~\cite{horodecki:RMP}:
indeed, classical communication cannot generate entanglement, which
hence becomes a physical resource that can be processed, but not
created.

In a resource theory of nonlocality, on the other hand, the LOCC
paradigm seems unjustified: even mere classical communication
constitutes in fact a nonlocal resource and, as such, cannot be
granted freely. For this reason, some authors consider the natural
operational paradigm of nonlocality to be that of \emph{local
  operations and shared randomness} (LOSR)~\cite{dukaric}. (A notable
exception to this argument occurs if nonlocality is measured in terms
of \emph{private} correlations: in this case, \emph{public} classical
communication can be freely allowed~\cite{acin:arxiv}.) In the LOSR
framework, separated parties are forbidden all sorts of communication,
being allowed though to ``synchronize'' their local operations with
respect to a common classical random variable shared in
advance. Hence, nonlocal correlations being defined as those
correlations that cannot be simulated by shared
randomness~\cite{short}, nonlocality naturally becomes a resource in
the LOSR paradigm.

The resource theory of quantum entanglement, with respect to the
resource theory of nonlocality, has received until now much more
attention in the literature: correspondingly, many results are known
about the interconversion of quantum states by LOCC
transformations~\cite{horodecki:RMP}, while much less is known about
the LOSR case~\cite{dukaric}. The aim of the present letter is to
contribute in bridging this gap, by providing a set of necessary and
sufficient conditions for the existence of an LOSR protocol
transforming one distributed quantum state into another. Such
conditions, rather than algebraic, are \emph{operational}, in the
sense that they are expressed in terms of the payoffs that a quantum
state yields in nonlocal games. More precisely, the main result of
this letter is to show that one quantum state can be transformed into
another by means of an LOSR protocol, if and only if the former yields
a higher payoff than the latter for a whole class of nonlocal games,
which we call \emph{semi-quantum} nonlocal games. A remarkable merit
of our analysis is to provide a simple and insightful proof of the
fact that \emph{all entangled quantum states are
  nonlocal}~\cite{popescu:hidden}: a corollary of our main result is
that any entangled quantum state yields a strictly higher payoff than
every separable state, in at least one semi-quantum nonlocal
game. This general fact will be also illustrated in an explicit
example, clarifying how semi-quantum nonlocal games are able to
faithfully witness entanglement.\medskip

\emph{Nonlocality ordering.}---In order to rigorously state the main
result (Prop.~\ref{prop:1} below), we first need to introduce some
notation and few definitions. In what follows, all quantum systems are
finite-dimensional (i.e. their Hilbert spaces, denoted by $\sH$, are
finite-dimensional) and index sets (denoted by $\outs$, $\outt$,
$\outx$, and $\outy$) contain only a finite number of elements. The
convex set of probability distributions defined on an index set
$\outx$ is denoted by $\prob(\outx)$. The set of linear operators
acting on a Hilbert space $\sH$ is denoted by $\linear(\sH)$. The set
of density matrices (i.e. positive semi-definite, trace-one operators)
is denoted by $\sS(\sH)\subset \linear(\sH)$.

A random source of states of a quantum system $A$ is represented by an
ensemble $\tau=(\{p(s),\tau^s\};s\in\outs)$, where $p\in\prob(\outs)$
and $\tau^s\in\sS(\sH_A)$, for all $s$. Given an outcome set
$\outx=\{x\}$ and a quantum system $A$ with Hilbert space $\sH_A$, an
$\outx$-probability operator-valued measure ($\outx$-POVM, for short)
on $A$ is a family $P=\left(P^x;x\in\outx\right)$ of positive
semi-definite operators $P^x\in\linear(\sH_A)$, such that
$\sum_{x\in\outx}P^x=\openone$. We denote by $\test{M}(A;\outx)$ the
convex set of all $\outx$-POVMs on $A$. A POVM $P\in\test{M}(A;\outx)$
induces, via the relation $p(x)=\Tr[P^x\rho]$, a linear function
$P:\rho\mapsto P\rho$ from $\sS(\sH_A)$ to $\prob(\outx)$. POVMs in
$\test{M}(A;\outx)$ are used to model measurements performed on a
quantum system $A$ with outcomes in $\outx$.

The notion of nonlocal games is of central importance in our
discussion (we begin here by considering the bipartite case; the
multipartite case follows directly and will be briefly discussed at
the end of the paper):

\begin{definition}\label{def:non-loc-games}
  The rules of a \emph{nonlocal game} $\game$ consist of the
  following: four index sets $\outs=\{s\}$, $\outt=\{t\}$,
  $\outx=\{x\}$, and $\outy=\{y\}$; two probability distributions
  $p\in\prob(\outs)$ and $q\in\prob(\outt)$; a payoff function
  $\wp:\outs \times\outt \times\outx \times\outy\to \mathbb{R}$. A
  referee picks indices $s\in\outs$ and $t\in\outt$ at random with
  probabilities $p(s)$ and $q(t)$, and sends them separately to two
  players, say Alice and Bob, respectively. The two players, without
  communicating with each other, must compute answers $x\in\outx$ and
  $y\in\outy$, respectively, and send them to the referee, who will
  then pay them both (i.e. the game is \emph{collaborative}) an amount
  equal to $\wp(s,t,x,y)$. (It is understood that a negative payoff
  means a loss, i.e. the players must pay the referee.)
\end{definition}

First, the players are told the rules of the game. Knowing the rules,
the players are allowed to agree on any strategy and to share any
possible (static) resource. Later on, the players and the referee
agree to begin the game, and, from that moment on, an implicit rule of
all nonlocal games forbids the players to communicate. According to
quantum theory then, anything the two players can do is to share a
bipartite quantum state $\rho_{AB}\in\sS(\sH_A\otimes\sH_B)$ and,
depending on the questions $s$ and $t$ they are presented, perform
independent measurements on $A$ and $B$ with values in $\outx$ and
$\outy$, respectively.

Imagine now that the state $\rho_{AB}$ shared between Alice and Bob is
\emph{fixed}. It is a well-defined question to ask ``how good'' is the
state $\rho_{AB}$ for playing a given nonlocal game $\game$. In order
to answer this question, it is convenient to use a mathematical model
in which the referee communicates her questions to Alice and Bob by
means of a quantum channel. This means that the referee, depending on
which questions $s\in\outs$ and $t\in\outt$ she picked, prepares two
auxiliary quantum systems $A_0$ and $B_0$, with dimensions $\dim
\sH_{A_0}\ge|\outs|$ and $\dim \sH_{B_0}\ge|\outt|$, in the
orthonormal states $\pi^s:=|s\>\<s|$ and $\pi^t:=|t\>\<t|$, and sends
them to Alice and Bob, respectively. We suppose that the states are
transmitted without noise. Since Alice and Bob exactly know which game
they are playing and which state they are sharing, the payoff they
expect to gain (on average) can be expressed by the following formula:
\begin{equation}\label{eq:gen-corr}
\begin{split}
  \wp^*\left(\rho_{AB};\game\right)&:=\\
  \max&\sum_{s,t,x,y}p(s)q(t)\wp(s,t,x,y)\mu(x,y|s,t),
\end{split}
\end{equation}
where $\mu(x,y|s,t)$ is the joint conditional probability distribution
computed as
\begin{equation}\nonumber
  \Tr\left[(P_{A_0A}^x\otimes
      Q_{BB_0}^y)(\pi^s_{A_0}\otimes\rho_{AB}\otimes\pi^t_{B_0})\right],
\end{equation}
and the maximization is performed over all POVMs
$P\in\test{M}(A_0A;\outx)$ and $Q\in\test{M}(BB_0;\outy)$.

The function $\wp^*(\rho_{AB};\game)$ in~(\ref{eq:gen-corr}) measures
the ``nonlocal utility'' of $\rho_{AB}$ in playing a nonlocal game
$\game$. Accordingly, if another state
$\sigma_{A'B'}\in\sS(\sH_{A'}\otimes\sH_{B'})$ is such that
$\wp^*\left(\sigma_{A'B'};\game\right)\le
\wp^*\left(\rho_{AB};\game\right)$, we say that $\rho_{AB}$ is better
than $\sigma_{A'B'}$ for playing $\game$. By extending this definition
to \emph{all} nonlocal games, we can introduce the following relation:

\begin{definition}\label{def:nonlocality}
  A bipartite state $\rho_{AB}\in\sS(\sH_A\otimes\sH_B)$ is said to be
  (\emph{definitely}) \emph{more nonlocal} than another bipartite
  state $\sigma_{A'B'}\in\sS(\sH_{A'}\otimes\sH_{B'})$, written
  $\rho_{AB}\supnl\sigma_{A'B'}$, if and only if
  $\wp^*\left(\rho_{AB};\game\right)\ge
  \wp^*\left(\sigma_{A'B'};\game\right)$, for all nonlocal games
  $\game$.
\end{definition}

The above definition can be equivalently reformulated in terms of Bell
inequalities~\cite{bell} as follows. Since it is known that to any
nonlocal game there corresponds a Bell inequality and, conversely, to
any Bell inequality there corresponds a nonlocal
game~\cite{bell-games}, we can equivalently say that
$\rho_{AB}\supnl\sigma_{A'B'}$, if and only if $\rho_{AB}$ appears to
be more nonlocal than $\sigma_{A'B'}$ with respect to all Bell
inequalities (or, more precisely speaking, all Bell
expressions~\cite{vertesi}).\medskip

\emph{Local operations and shared randomness.}---Let us now turn to
the LOSR paradigm within quantum theory (again, we begin with the
bipartite case): a completely positive trace-preserving (CPTP) map
$\mE:\linear(\sH_A\otimes\sH_{B})\to\linear(\sH_{A'}\otimes\sH_{B'})$
is said to be an LOSR transformation, if it can be written as
$\sum_i\nu(i)\mE^i\otimes\mF^i$, where
$\mE^i:\linear(\sH_A)\to\linear(\sH_{A'})$ and $\mF^i:\linear(\sH_{B})
\to \linear(\sH_{B'})$ are CPTP maps for all $i$, and $\nu(i)$ is a
probability distribution~\cite{causal}. We then introduce the
following definition:

\begin{definition}\label{def:rand-suff}
  A bipartite state $\rho_{AB}\in\sS(\sH_A\otimes\sH_B)$ is said to be
  \emph{LOSR sufficient} for another bipartite state
  $\sigma_{A'B'}\in\sS(\sH_{A'}\otimes\sH_{B'})$, written
  $\rho_{AB}\succlosr\sigma_{A'B'}$, if and only if there exists an
  LOSR transformation mapping $\rho_{AB}$ into $\sigma_{A'B'}$.
\end{definition}

It is a rather straightforward exercise to prove that the relation
$\succlosr$ implies the relation $\supnl$. In fact,
$\rho_{AB}\supnl(\mE^i\otimes\mF^i)\rho_{AB}$ trivially holds for all
$i$. On the other hand, the payoff achievable with the convex
combination $\sum_i\nu(i)(\mE^i\otimes\mF^i)\rho_{AB}$ cannot exceed
the best payoff achievable with each of its component, i.e. there
exists $i$ such that $(\mE^i\otimes\mF^i)\rho_{AB}\supnl
\sum_i\nu(i)(\mE^i\otimes\mF^i)\rho_{AB}$. This proves the claim.

It is also straightforward to prove that separable states are the
endpoints of the relation $\succlosr$, i.e. for any separable state
$\sigma_{A'B'}$, $\rho_{AB}\succlosr \sigma_{A'B'}$, for all
$\rho_{AB}$. Suppose, in fact, that
$\sigma_{A'B'}\in\sS(\sH_{A'}\otimes\sH_{B'})$ is a separable state,
i.e., $\sigma_{A'B'}=\sum_i\nu(i)\gamma^i_{A'}\otimes\chi^i_{B'}$, for
some probability distribution $\nu(i)$ and some local states
$\gamma^i\in\sS(\sH_{A'})$ and $\chi^i\in\sS(\sH_{B'})$. Then, there
always exists a ``discard-and-prepare'' LOSR map
$\mE:\linear(\sH_{A}\otimes\sH_{B})\to\linear(\sH_{A'}\otimes\sH_{B'})$
such that $\sigma_{A'B'}=\mE(\rho_{AB})$, for all
$\rho_{AB}\in\sS(\sH_{A}\otimes\sH_{B})$, proving the claim.

These two facts together make it easy to verify that, in any nonlocal
game $\game$, all separable states yield exactly the same payoff
$\wp_{\operatorname{sep}}(\game)$. This remark will be useful in what
follows.\medskip

\emph{Semi-quantum nonlocal games.}---At this point, the question of
whether the implication can be reversed, i.e. whether the relation
$\supnl$ implies $\succlosr$ or not, naturally arises, and its answer
is ``no''. Let us consider in fact those entangled quantum states
(called LHVPOV states~\cite{barret}) for which a local-hidden-variable
model exists, describing the outcome statistics of every local POVM
measurement performed on them. This means that, for any nonlocal game
$\game$, the expected payoff obtainable from such entangled states
never exceeds that obtainable from separable states. However, it is
impossible to create an entangled state (even if LHVPOV) by acting
with LOSR transformations on separable states. This proves the claim
that $\supnl$ does not imply $\succlosr$.

The relation $\supnl$ is too weak to imply $\succlosr$. We hence
introduce a stronger version of $\supnl$, by suitably enlarging the
set of nonlocal games we consider. The extended notion of nonlocal
games we need is the following:

\begin{definition}\label{def:extgames}
  The rules of a \emph{semi-quantum nonlocal game} $\extgame$ consist
  of: four index sets $\outs=\{s\}$, $\outt=\{t\}$, $\outx=\{x\}$, and
  $\outy=\{y\}$; two quantum systems $A_0$ and $B_0$; two random
  sources $\tau=(\{p(s),\tau^s\};s\in\outs)$ and
  $\omega=(\{q(t),\omega^t\};t\in\outt)$ on $A_0$ and $B_0$,
  respectively; a payoff function
  $\wp:\outs\times\outt\times\outx\times\outy\to\mathbb{R}$. A referee
  picks indices $s\in\outs$ and $t\in\outt$ at random with
  probabilities $p(s)$ and $q(t)$, and sends the corresponding states
  $\tau^s$ and $\omega^t$ to Alice and Bob, respectively (without
  revealing the actual indices $s$ and $t$ though).  The two players,
  without communicating with each other, must compute answers
  $x\in\outx$ and $y\in\outy$, respectively, and send them to the
  referee, who will then pay them both an amount equal to
  $\wp(s,t,x,y)$.
\end{definition}

In other words, while in conventional nonlocal games the referee asks
the players ``classical'' questions, in semi-quantum nonlocal games
the referee is allowed to ask them ``quantum'' questions. Clearly,
semi-quantum nonlocal games contain, as special cases, conventional
nonlocal games (Def.~\ref{def:non-loc-games}), whenever the states
that the referee sends to Alice and Bob are perfectly distinguishable,
i.e. ``classical''. The situation is depicted in
Figure~\ref{fig:extended}.

\begin{figure}[h]
  \begin{center}
    \includegraphics[width=6cm]{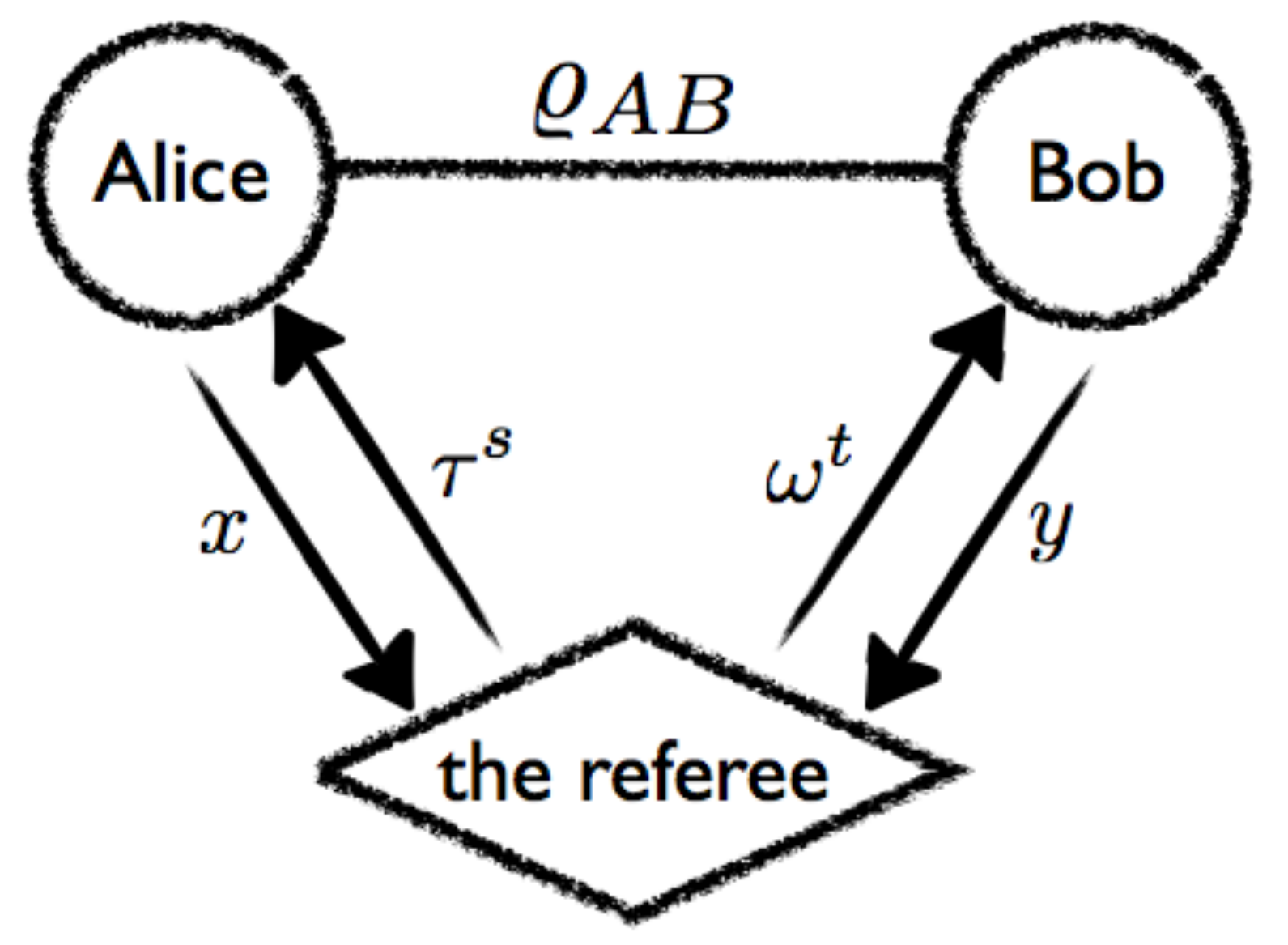}
  \end{center}
  \caption{\small In a semi-quantum nonlocal game
    (Def.~\ref{def:extgames}), while players still reply with
    ``classical'' answers, the referee is allowed to ask ``quantum''
    questions. Whenever the signals $\tau^s$ and $\omega^t$ are
    perfectly distinguishable, i.e. classical, the case of
    conventional nonlocal games (Def.~\ref{def:non-loc-games}) is
    recovered. By means of semi-quantum nonlocal games, it is possible
    to show that all entangled quantum states are nonlocal
    (Cor.~\ref{cor:all-sep-st}).}
  \label{fig:extended}
\end{figure}

As in the case of conventional nonlocal games, the two players are
allowed to share a bipartite quantum state, say $\rho_{AB}$, so that
the expected payoff $\wp^*\left(\rho_{AB};\extgame\right)$ is given by
the same formula~(\ref{eq:gen-corr}), the only difference being that
the joint conditional probability distribution $\mu(x,y|s,t)$ is now
computed as
\begin{equation*}
  \Tr\left[(P_{A_0A}^x\otimes
    Q_{BB_0}^y)(\tau^s_{A_0}\otimes\rho_{AB}\otimes\omega^t_{B_0})\right].
\end{equation*}
Analogously to what was done before, we can compare the nonlocal
utility of two quantum states for \emph{all} semi-quantum nonlocal
games and introduce the following relation:

\begin{definition}\label{def:strong-nonlocality}
  Given two bipartite states $\rho_{AB}\in\sS(\sH_A\otimes\sH_B)$ and
  $\sigma_{A'B'}\in\sS(\sH_{A'}\otimes\sH_{B'})$, we define the
  relation $\rho_{AB}\sups\sigma_{A'B'}$, meaning that
  $\wp^*\left(\rho_{AB};\extgame\right)\ge
  \wp^*\left(\sigma_{A'B'};\extgame\right)$, for all semi-quantum
  nonlocal games $\extgame$.
\end{definition}

Since semi-quantum nonlocal games contain conventional nonlocal games
as a special case, the relation $\sups$ implies the relation
$\supnl$. Moreover, along the same line of thoughts used above to show
that $\succlosr$ implies $\supnl$, it is straightforward to prove that
$\succlosr$ also implies $\sups$.\medskip

\emph{A fundamental equivalence.}---We are now ready to state the main
result of this letter:

\begin{proposition}\label{prop:1}
  Given two bipartite states $\rho_{AB}$ and $\sigma_{A'B'}$,
  $\rho_{AB}\sups\sigma_{A'B'}$ if and only if $\rho_{AB}\succlosr\sigma_{A'B'}$.
\end{proposition}

The proof of Prop.~\ref{prop:1} is based on arguments very similar to
those used in Ref.~\cite{qblack}, and crucially uses the Separation
Theorem between convex sets~\cite{separation}. Being rather technical
in nature, we omit it here, pointing the interested reader to the
supplemental material accompanying this letter~\cite{appendix}. Here
we only discuss one important consequence of our main result, that is,
Prop.~\ref{prop:1} implies that \emph{any entangled state is strictly
  more nonlocal than every separable state}, as stated in the
following corollary:

\begin{corollary}\label{cor:all-sep-st}
  In any semi-quantum nonlocal game $\extgame$, all separable quantum
  states yield exactly the same payoff
  $\wp_{\operatorname{sep}}(\extgame)$. Moreover, a quantum state
  $\rho_{AB}$ is entangled if and only if there exists a semi-quantum
  nonlocal game $\extgame$, for which
  $\wp^*(\rho_{AB};\extgame)>\wp_{\operatorname{sep}}(\extgame)$.
\end{corollary}

In other words, any entangled quantum state has a form of nonlocality,
which is ``hidden''~\cite{popescu:hidden} for conventional nonlocal
games (and hence Bell inequalities), but becomes apparent when playing
semi-quantum nonlocal games. The proof of the above corollary is a
direct consequence of the fact that separable states, being the
endpoints of the relation $\succlosr$, are also the endpoints of the
relation $\sups$, due to the equivalence established by
Prop.~\ref{prop:1}.\medskip

\emph{An example.}---In order to illustrate the superiority of
semi-quantum nonlocal games, with respect to conventional ones, in
witnessing entanglement, we describe now the example of a semi-quantum
nonlocal game, in which every entangled state gives rise to joint
question-answer probability distributions that cannot be explained as
coming from \emph{any} separable state, even if supplemented with an
unlimited amount of shared randomness. This is true also for entangled
LHVPOV states, which are instead completely indistinguishable from
separable states, if only conventional nonlocal games
(Def.~\ref{def:non-loc-games}) are considered. The example, directly
stemming from the proof of Proposition~\ref{prop:1}
(see~\cite{appendix}), is described here only in the case of two-qubit
states; we remark, however, that the same construction can be easily
carried over to any finite dimensional case.

In our example, $\outs=\outt=\outx=\outy=\{1,2,3,4\}$, the auxiliary
quantum systems used by the referee to encode her questions are
represented by two qubits, i.e.
$\sH_{A_0}\cong\sH_{B_0}\cong\mathbb{C}^2$, and the ``question
states'' are the four tetrahedral states $|\psi^1\>$, $|\psi^2\>$,
$|\psi^3\>$, and $|\psi^4\>$ defined by Davies~\cite{davies}. Notice
that the choice of the question states is somewhat arbitrary: the
important point is that their density matrices constitute a basis for
the linear space $\linear(\mathbb{C}^2)$. (The definition of the
probability distributions on $\outs$ and $\outt$, as well as that of
the payoff function, are not necessary for our argument and can be
omitted.)

Given a two-qubit state $\rho_{AB}$, let us consider the joint
conditional question-answer probability distribution
$\bar\mu(x,y|s,t\|\rho)$ computed as
\begin{equation}\label{eq:mubar}
\Tr\left[(B^x_{A_0A}\otimes B^y_{BB_0})(\psi^s_{A_0}\otimes\rho_{AB}\otimes\psi^t_{B_0})\right],
\end{equation}
where $B^1, B^2, B^3, B^4$ are, respectively, the four orthogonal Bell
measurements on $\Phi^+, \Phi^-, \Psi^+, \Psi^-$. In the process of
proving Prop.~\ref{prop:1} (see~\cite{appendix}), it is also shown
that, in particular, the two-qubit state $\rho_{AB}$ is entangled if
and only if, for any (possibly higher-dimensional) separable state
$\sigma_{A'B'}$ and for any possible POVMs $P\in\test{M}(A_0A';\outx)$
and $Q\in\test{M}(B'B_0;\outy)$,
\begin{equation*}
\begin{split}
  &\bar\mu(x,y|s,t\|\rho)\neq\\
  &\Tr\left[(P^x_{A_0A'}\otimes
    Q^y_{B'B_0})(\psi^s_{A_0}\otimes\sigma_{A'B'}\otimes \psi^t_{B_0})\right].
\end{split}
\end{equation*}
In fact, one can easily check, following the proof of
Prop.~\ref{prop:1}, that an equality in the above equation, for some
separable state $\sigma_{A'B'}$ and some POVMs $P$ and $Q$, would
imply the existence of an LOSR transformation mapping $\sigma_{A'B'}$
into $\rho_{AB}$, hence leading to a contradiction, due to the fact
that LOSR transformations cannot create entangled states from
separable ones. In other words, the state $\rho_{AB}$ is entangled if
and only if the joint conditional probability distribution
$\bar\mu(x,y|s,t\|\rho)$, computed in Eq.~(\ref{eq:mubar}), is out of
reach for any possible separable state, even with the help of
unlimited shared randomness (represented here by the possibility of
$\sigma_{A'B'}$ being on a higher-dimensional Hilbert space).\medskip

\emph{Multipartite states.}---Before concluding, we remark here that
our approach can be straightforwardly extended to consider
multipartite LOSR transformations $\mE: \linear(\sH_A\otimes
\sH_B\otimes \sH_C\otimes\cdots) \to \linear(\sH_{A'}\otimes
\sH_{B'}\otimes \sH_{C'}\otimes\cdots)$ of the form
$\mE=\sum_i\nu(i)\mE^i\otimes \mF^i \otimes \mathcal{G}^i
\otimes\cdots$, where $\mE^i: \linear(\sH_A)\to \linear(\sH_{A'})$,
$\mF^i: \linear(\sH_B)\to \linear(\sH_{B'})$, $\mathcal{G}^i:
\linear(\sH_C)\to \linear(\sH_{C'})$, and so on, are all CPTP maps,
for all $i$. This can be done by considering multipartite semi-quantum
nonlocal games, in which all the players independently receive their
``quantum questions'' from the referee, and by following the same
arguments used to prove the bipartite case.\medskip

\emph{Conclusions.}---We showed that one quantum state can be
transformed into another by means of an LOSR protocol, if and only if
the former is ``more nonlocal'' than the latter, where nonlocality is
quantified by means of semi-quantum nonlocal games
(Def.~\ref{def:strong-nonlocality}). As a by-product, we obtained a
clear-cut proof that \emph{any entangled quantum state is always
  nonlocal}, a fact that should be considered in light of previous
works reaching the same conclusion, although from very different
routes~\cite{popescu:hidden}. In order to support our analysis and
show the superiority of semi-quantum nonlocal games, with respect to
conventional ones, in witnessing entanglement, we also provided an
explicit example of a semi-quantum nonlocal game, in which any
entangled state gives rise to joint question-answer probability
distributions that cannot be explained classically, even if an
unlimited amount of shared randomness is granted.\medskip

\emph{Acknowledgments.}---The author is grateful to Denis Rosset and
Mark M. Wilde for pointing out mistakes in a previous version. An
exchange with Antonio Ac\'in and Miguel Navascu\'es is also gratefully
acknowledged. This research was supported by the Program for
Improvement of Research Environment for Young Researchers from Special
Coordination Funds for Promoting Science and Technology (SCF)
commissioned by the Ministry of Education, Culture, Sports, Science
and Technology (MEXT) of Japan.

\appendix


\onecolumngrid
\newpage

\setcounter{page}{1}
\renewcommand{\thepage}{All Entangled Quantum States
  Are Nonlocal: Supplemental Material -- \arabic{page}/2}

\section{Supplemental Material}

\noindent{\bf Remark.} The numbering of equations and references
follows that given in the main text.\bigskip

\noindent {\bf Remark.} In what follows, for notational convenience,
the density matrices $\tau^s$ and $\omega^t$ are taken sub-normalized,
so that $\Tr[\tau^s]=p(s)$ and $\Tr[\omega^t]=q(t)$.\bigskip

\begin{proof}[Proof of Proposition~\ref{prop:1}]
  We explicitly prove only the non-trivial direction, i.e. the ``only
  if'' part of the statement.\medskip

  We start by making the following observation: the payoff function
  $\wp^*(\rho_{AB};\extgame)$ contains a maximization over local
  measurements $P\in\test{M}(A_0A;\outx)$ and
  $Q\in\test{M}(BB_0;\outy)$ on Alice's and Bob's systems,
  respectively. The set of local measurements does not constitute a
  convex set, in the sense that a convex combination $p(P'\otimes
  Q')+(1-p)(P''\otimes Q'')$, for $P',P''\in\test{M}(A_0A;\outx)$ and
  $Q',Q''\in\test{M}(BB_0;\outy)$, in general cannot be written as
  $P\otimes Q$, for any $P\in\test{M}(A_0A;\outx)$ and
  $Q\in\test{M}(BB_0;\outy)$. However, since the function
\begin{equation}\nonumber
  g(\rho_{AB};\extgame;P,Q):=\sum_{s,t,x,y}\wp(s,t,x,y)\
  \Tr\left[\left(P_{A_0 A}^x\otimes Q_{BB_0}^y\right)\
    \left(\tau_{A_0}^s\otimes\rho_{AB}\otimes\omega_{B_0}^t\right)\right]
\end{equation}
is linear in the POVMs $P$ and $Q$, we can extend it by linearity to
any convex combination $\sum_i \nu(i)P^x_{A_0 A}(i)\otimes Q^y_{BB_0}(i)$, where
$\nu(i)$ are probabilities and $P(i)\in\test{M}(A_0A;\outx)$ and
$Q(i)\in\test{M}(BB_0;\outy)$, for all $i$. Let us denote by
$\conv\{\test{M}(A_0A;\outx)\otimes\test{M}(BB_0;\outy)\}$ the set of
such convex combinations of local POVMs.

Since a linear function is, in particular, convex; since a convex
function on a convex set achieves its maximum on the extremal points
of such set; and since the extremal points of
$\conv\{\test{M}(A_0A;\outx)\otimes\test{M}(BB_0;\outy)\}$ are, by
construction, local POVMs, we have that
\begin{equation}\label{eq:conv-opt}
  \max_{Z\in \conv\{\test{M}(A_0A;\outx)\otimes\test{M}(BB_0;\outy)\}}g(\rho_{AB};\extgame;Z)=\wp^*(\rho_{AB};\extgame).
\end{equation}

For any choice of $\outs,\outt,\outx,\outy,A_0,B_0,\tau,\omega$ (the
meaning of the notation is the same as in
Def.~\ref{def:extgames}), let us now consider the set of
probability distributions defined as follows :
\begin{equation}\nonumber
\begin{split}
  &\prob(\rho_{AB};\outs,\outt,\outx,\outy,A_0,B_0,\tau,\omega):=\\
&\left\{\mu\in\prob(\outs\times\outt\times\outx\times\outy)\left|
\begin{split}
&\mu(s,t,x,y)=\Tr\left[Z_{A_0 ABB_0}^{x,y}\
      \left(\tau_{A_0}^s\otimes\rho_{AB}\otimes\omega_{B_0}^t\right)\right],\\
    &Z\in \conv\{\test{M}(A_0A;\outx)\otimes\test{M}(BB_0;\outy)\}
\end{split}
\right.\right\}.
\end{split}
\end{equation}
Due to the identity~(\ref{eq:conv-opt}), we have that
\begin{equation}\nonumber
  \wp^*(\rho_{AB};\extgame)=\max_{\mu\in \prob(\rho_{AB};\outs,\outt,\outx,\outy,A_0,B_0,\tau,\omega)}\sum_{s,t,x,y}\mu(s,t,x,y)\wp(s,t,x,y).
\end{equation}

The crucial point, now, is that, by construction, the set
$\prob(\rho_{AB};\outs,\outt,\outx,\outy,A_0,B_0,\tau,\omega)$ is
convex, as it inherits the convex structure from
$\conv\{\test{M}(A_0A;\outx)\otimes\test{M}(BB_0;\outy)\}$. Therefore,
following the same arguments presented in more detail
in~\cite{qblack}, as a consequence of the so-called ``separation
theorem'' for convex sets~\cite{separation},
Def.~\ref{def:strong-nonlocality} can be reformulated in the
following way:
\begin{equation}\nonumber
  \rho_{AB}\sups\sigma_{A'B'}\quad \Leftrightarrow\quad 
\prob(\rho_{AB};\outs,\outt,\outx,\outy,A_0,B_0,\tau,\omega)
  \supseteq
  \prob(\sigma_{A'B'};\outs,\outt,\outx,\outy,A_0,B_0,\tau,\omega),
\end{equation}
for any choice of $\outs,\outt,\outx,\outy,A_0,B_0,\tau,\omega$.

More explicitly stated, $\rho_{AB}\sups\sigma_{A'B'}$ if and only if,
for any choice of $\outs,\outt,\outx,\outy,A_0,B_0,\tau,\omega$, and
for any POVM $Z\in
\conv\{\test{M}(A_0A';\outx)\otimes\test{M}(B'B_0;\outy)\}$, there
exists a POVM $\bar Z\in
\conv\{\test{M}(A_0A;\outx)\otimes\test{M}(BB_0;\outy)\}$, such that
\begin{equation}\label{eq:strong-nonlocality2}
\Tr\left[\bar Z_{A_0
        ABB_0}^{x,y}\
    \left(\tau_{A_0}^s\otimes\rho_{AB}\otimes\omega_{B_0}^t\right)\right]=
  \Tr\left[Z_{A_0
        A'B'B_0}^{x,y}\
    \left(\tau_{A_0}^s\otimes\sigma_{A'B'}\otimes\omega_{B_0}^t\right)\right],
\end{equation}
for all $s\in\outs$, $t\in\outt$, $x\in\outx$, and $y\in\outy$.

Let us now choose $A_0$ and $B_0$ to be such that
$\sH_{A_0}\cong\sH_{A'}$ and $\sH_{B_0}\cong\sH_{B'}$. Moreover, let
us introduce two further auxiliary quantum systems $A_1$ and $B_1$,
with $\sH_{A_1}\cong \sH_{A_0}$ ($\cong \sH_{A'}$) and
$\sH_{B_1}\cong\sH_{B_0}$ ($\cong\sH_{B'}$). Next, let us choose
$(\tau^s;s\in\outs)$ on $A_0$ to be given by
\begin{equation}\nonumber
\tau_{A_0}^s=\Tr_{A_1}\left[\left(\Theta_{A_1}^s\otimes \openone_{A_0}\right)\ \Psi^+_{A_1A_0}\right],
\end{equation}
and $(\omega^t;t\in\outt)$ on $B_0$ by
\begin{equation}\nonumber
\omega_{B_0}^t=\Tr_{B_1}\left[\left(\openone_{B_0}\otimes\Upsilon_{B_1}^t\right)\ \Psi^+_{B_0B_1}\right],
\end{equation}
where $\Psi^+$ denotes a maximally entangled state and the two POVMs
$\Theta\in\test{M}(A_1;\outs)$ and $\Upsilon\in\test{M}(B_1;\outt)$
are both informationally complete (i.e. their linear span coincide
with $\linear(\sH_{A_1})$ and $\linear(\sH_{B_1})$,
respectively). Then, Eq.~(\ref{eq:strong-nonlocality2}) can be written
as
\begin{equation}\nonumber
  \begin{split}
  &\Tr\left[\left(\Theta^s_{A_1}\otimes\bar Z_{A_0 ABB_0}^{x,y}\otimes\Upsilon^t_{B_1}\right)\
    \left(\Psi^+_{A_1A_0}\otimes\rho_{AB}\otimes\Psi^+_{B_0B_1}\right)\right]\\
=&\Tr\left[\left(\Theta^s_{A_1}\otimes Z_{A_0
        A'B'B_0}^{x,y}\otimes\Upsilon^t_{B_1}\right)\
    \left(\Psi^+_{A_1A_0}\otimes\sigma_{A'B'}\otimes\Psi^+_{B_0B_1}\right)\right],
\end{split}
\end{equation}
for all $s,t,x,y$.

Due to the fact that the POVMs $\Theta$ and $\Upsilon$ have been
chosen to be informationally complete, we arrive at the following
conclusion: if $\rho_{AB}\sups\sigma_{A'B'}$, then, for any choice of
outcome sets $\outx,\outy$ and POVMs $P\in\test{M}(A_0A';\outx)$ and
$Q\in\test{M}(B'B_0;\outy)$, there exists a POVM $\bar
Z\in\conv\{\test{M}(A_0A;\outx) \otimes \test{M}(BB_0;\outy)\}$, such
that
\begin{equation}\label{eq:strong-nonlocality3}
  \begin{split}
  &\Tr_{A_0ABB_0}\left[\left(\openone_{A_1}\otimes\bar Z_{A_0 ABB_0}^{x,y}\otimes\openone_{B_1}\right)\
    \left(\Psi^+_{A_1A_0}\otimes\rho_{AB}\otimes\Psi^+_{B_0B_1}\right)\right]\\
=&\Tr_{A_0A'B'B_0}\left[\left(\openone_{A_1}\otimes P_{A_0
        A'}^x\otimes Q_{B'B_0}^y\otimes\openone_{B_1}\right)\
    \left(\Psi^+_{A_1A_0}\otimes\sigma_{A'B'}\otimes\Psi^+_{B_0B_1}\right)\right],
\end{split}
\end{equation}
for all $x\in\outx$ and all $y\in\outy$.

Let us now choose $\outx$ and $\outy$ such that
$|\outx|=(\dim\sH_{A'})^2$ and $|\outy|=(\dim\sH_{B'})^2$, and the
POVMs $P$ and $Q$ to be the generalized Bell measurements on $A_0A'$
and $B'B_0$, respectively. With this choice in mind, let us denote the
right-hand side of~(\ref{eq:strong-nonlocality3}) by
$\sigma_{A_1B_1}^{x,y}$. The protocol of quantum teleportation provides
unitary operators $U^x:\sH_{A_1}\to\sH_{A'}$ and
$V^y:\sH_{B_1}\to\sH_{B'}$ such that
\begin{equation}\nonumber
\sum_{x\in\outx}\sum_{y\in\outy}\left(U^x_{A_1}\otimes
    V^y_{B_1}\right) \sigma^{x,y}_{A_1B_1}\left(U^x_{A_1}\otimes
    V^y_{B_1}\right)^\dag=\sigma_{A'B'}.
\end{equation}

On the other hand, since $\rho_{AB}\sups\sigma_{A'B'}$, via
equation~(\ref{eq:strong-nonlocality3}), we know that there exists a
POVM $\bar Z\in\conv\{\test{M}(A_0A;\outx) \otimes
\test{M}(BB_0;\outy)\}$, such that
\begin{equation}\label{eq:strong-nonlocality4}
\begin{split}
&\sigma_{A'B'}\\
=&\sum_{x,y}\left(U^x_{A_1}\otimes
  V^y_{B_1}\right) \Tr_{A_0ABB_0}\left[\left(\openone_{A_1}\otimes
    \bar Z_{A_0
        ABB_0}^{x,y}\otimes\openone_{B_1}\right)\
    \left(\Psi^+_{A_1A_0}\otimes\rho_{AB}\otimes\Psi^+_{B_0B_1}\right)\right]\left(U^x_{A_1}\otimes
    V^y_{B_1}\right)^\dag.
\end{split}
\end{equation}

Finally, by expanding the POVM elements $\bar Z_{A_0 ABB_0}^{x,y}$
into a convex combination $\bar Z_{A_0 ABB_0}^{x,y}=\sum_i \nu(i)\bar
P_{A_0A}^x(i)\otimes\bar Q_{BB_0}^y(i)$, where $\bar
P(i)\in\test{M}(A_0A;\outx)$ and $\bar Q(i)\in\test{M}(BB_0;\outy)$
for all $i$, and by defining CPTP maps $\mE^i(z):\linear(\sH_{A})\to
\linear(\sH_{A'})$ and $\mF^i(w):\linear(\sH_{B})\to\linear(\sH_{B'})$
as
\begin{equation}\nonumber
  \mE^i(z_A):=\sum_{x\in\outx}U^x_{A_1}\ \Tr_{A_0A}\left[\left(\openone_{A_1}\otimes
      \bar P_{A_0
        A}^x(i)\right)\ \left(\Psi^+_{A_1A_0}\otimes z_A\right)\right]\left(U^x_{A_1}\right)^\dag
\end{equation}
and
\begin{equation}\nonumber
  \mF^i(w_B):=\sum_{y\in\outy}V^y_{B_1}\ \Tr_{B_0B}\left[\left(      \bar Q_{B
        B_0}^y(i)\otimes \openone_{B_1}\right)\ \left(w_B\otimes \Psi^+_{B_0B_1}\right)\right]\left(V^y_{B_1}\right)^\dag,
\end{equation}
Eq.~(\ref{eq:strong-nonlocality4}) can be rewritten as
$\sigma_{A'B'}=\sum_i\nu(i)(\mE^i_{A}\otimes\mF^i_{B})(\rho_{AB})$. This
concludes the proof.
\end{proof}

\end{document}